\documentclass[sigconf]{aamas} 

\usepackage{balance} 

\usepackage{xspace}
\usepackage{amsmath}
\usepackage{cleveref}
\usepackage{fancyvrb}
\usepackage{mathtools}
\usepackage{tikz}
\usepackage{todonotes}
\usepackage[noend]{algpseudocode}
\usepackage{algorithm}
\usetikzlibrary{matrix,decorations.pathreplacing,tikzmark}
\usepackage{subcaption}

\setcopyright{ifaamas}
\acmConference[AAMAS '23]{Proc.\@ of the 22nd International Conference
on Autonomous Agents and Multiagent Systems (AAMAS 2023)}{May 29 -- June 2, 2023}
{London, United Kingdom}{A.~Ricci, W.~Yeoh, N.~Agmon, B.~An (eds.)}
\copyrightyear{2023}
\acmYear{2023}
\acmDOI{}
\acmPrice{}
\acmISBN{}

\acmSubmissionID{781}

\title[Computing the Best Policy That Survives a Vote]{Computing the Best Policy That Survives a Vote}

\author{Andrei Constantinescu}
\affiliation{
  \institution{ETH Zurich}
  \city{Zurich}
  \country{Switzerland}}
\email{aconstantine@ethz.ch}

\author{Roger Wattenhofer}
\affiliation{
  \institution{ETH Zurich}
  \city{Zurich}
  \country{Switzerland}}
\email{wattenhofer@ethz.ch}

\begin{abstract}
An assembly of $n$ voters needs to decide on $t$ independent binary issues. Each voter has opinions about the issues, given by a $t$-bit vector. Anscombe's paradox shows that a policy following the majority opinion in each issue may not survive a vote by the very same set of $n$ voters, i.e., more voters may feel unrepresented by such a majority-driven policy than represented. A natural resolution is to come up with a policy that deviates a bit from the majority policy but no longer gets more opposition than support from the electorate.
We show that a Hamming distance to the majority policy of at most $\floor*{(t - 1) / 2}$ can always be guaranteed, by giving a new probabilistic argument relying on structure-preserving symmetries of the space of potential policies.
Unless the electorate is evenly divided between the two options on all issues, we in fact show that a policy strictly winning the vote exists within this distance bound. 
Our approach also leads to a deterministic polynomial-time algorithm for finding policies with the stated guarantees, answering an open problem of previous work. For odd $t$, unless we are in the pathological case described above, we also give a simpler and more efficient algorithm running in expected polynomial time with the same guarantees.
We further show that checking whether distance strictly less than $\floor*{(t - 1) /2}$ can be achieved is NP-hard, and that checking for distance at most some input $k$ is FPT with respect to several natural parameters.
\end{abstract}


\newcommand{\BibTeX}{\rm B\kern-.05em{\sc i\kern-.025em b}\kern-.08em\TeX}

\DeclarePairedDelimiter{\floor}{\lfloor}{\rfloor}
\DeclarePairedDelimiter{\ceil}{\lceil}{\rceil}

\newcommand{\B}{\mathbb{B}}
\newcommand{\E}{\mathbb{E}}
\renewcommand{\P}{\mathbb{P}}  
\newcommand{\Z}{\mathbb{Z}}
\newcommand{\calD}{\mathcal{D}}
\newcommand{\calP}{\mathcal{P}}
\newcommand{\Qm}{\text{`?'}} 

\newcommand{\IWM}{{\sc IWM}\xspace}
\newcommand{\MajSup}{{\sc Win-Or-Tie-Prop}\xspace}
\newcommand{\UnaSup}{{\sc Una-Win-Prop}\xspace}

\newcommand{\Mod}[1]{\ (\mathrm{mod}\ #1)}

\newtheorem{theorem}{Theorem}
\newtheorem{lemma}[theorem]{Lemma}

\newtheorem{corollary}[theorem]{Corollary}

\newcommand{\drawcell}[2]{\draw[line width = 2.5mm, opacity = 0.2, color = {#2}] ({#1}.west) -- ({#1}.east);}
\newcommand{\drawline}[3]{\draw[line width = 2.5mm, opacity = 0.2, color = {#3}] ({#1}.west) -- ({#2}.east);}

\newcommand{\czero}{blue}
\newcommand{\cone}{red}

\newif\ifcomplver
\complverfalse

\begin{document}


\pagestyle{fancy}
\fancyhead{}


\maketitle 


\section{Introduction}

Once upon a time there was a country that was ruled by a good queen. The queen always asked her citizens about their opinions on the current issues in the country, in the form of yes-no questions. Then, she would form a policy following the majority opinion on each individual issue. Some of her people would be unhappy with the policy, because they disagreed with the policy more than they agreed. As long as these unhappy citizens were in a minority, the queen would not mind.

However, one day the queen noticed that the popular opinion policy would not be supported by a majority (e.g., Fig.~\ref{fig:anscombe_small}). Can the queen come up with a policy which is going to be supported by a majority of the people, while still respecting the opinion of the people as much as possible?
\begin{figure}[t]
\centering
\begin{tabular}{c|ccc}
          & $i_1$ & $i_2$ & $i_3$ \\
    \hline
    $v_1$ & \tikzmarknode{1_1}{1} & \tikzmarknode{1_2}{0} & \tikzmarknode{1_3}{0} \\
    $v_2$ & \tikzmarknode{2_1}{0} & \tikzmarknode{2_2}{1} & \tikzmarknode{2_3}{0} \\
    $v_3$ & \tikzmarknode{3_1}{0} & \tikzmarknode{3_2}{0} & \tikzmarknode{3_3}{1} \\
    $v_4$ & \tikzmarknode{4_1}{1} & \tikzmarknode{4_2}{1} & \tikzmarknode{4_3}{1} \\
    $v_5$ & \tikzmarknode{5_1}{1} & \tikzmarknode{5_2}{1} & \tikzmarknode{5_3}{1} 
\end{tabular}
\begin{tikzpicture}[overlay,remember picture, shorten >=-3pt, shorten <= -3pt] 
\drawcell{1_1}{\cone}
\drawcell{1_2}{\czero}
\drawcell{1_3}{\czero}
\drawcell{2_1}{\czero}
\drawcell{2_2}{\cone}
\drawcell{2_3}{\czero}
\drawcell{3_1}{\czero}
\drawcell{3_2}{\czero}
\drawcell{3_3}{\cone}
\drawcell{4_1}{\cone}
\drawcell{4_2}{\cone}
\drawcell{4_3}{\cone}
\drawcell{5_1}{\cone}
\drawcell{5_2}{\cone}
\drawcell{5_3}{\cone}
\end{tikzpicture}
\caption{Anscombe's Paradox. Five voters (rows) express their opinions on three equally important issues (columns). On each issue, opinion $1$ is the majority opinion, so it would be natural to propose the policy $111$. However, policy $111$ would not survive a vote, since the first three voters would vote against the policy. Example due to \cite{ostrogorksi_example}.}
\label{fig:anscombe_small}
\Description{Anscombe's Paradox}
\end{figure}%

We believe that this problem is relevant in various natural contexts. For example, a board of directors might map out the complete matrix on how each director thinks about each current issue. Together they want to decide on the best policy. As much as possible, this policy should reflect the majority opinion on each issue, but it should also be agreeable to the board as a whole, in that it will have not more opposition than support.

More formally, we model policy proposals and voters' preferences using $t$-bit vectors, also known as approval ballots \cite{approval_book}. A voter supports a proposal if the Hamming distance between their vector and the policy vector is less than $t / 2$. The voter opposes a proposal if the distance is more than $t / 2$. If $t$ is even, and the distance is exactly $t / 2$, the voter abstains from voting. An issue-wise majority proposal is formed by taking the majoritarian opinion on each issue, breaking ties arbitrarily. Anscombe's Paradox \cite{anscombe}, illustrated in Figure \ref{fig:anscombe_small}, shows that issue-wise majority might not survive a final vote, even with no abstentions and each issue having a strict majority for one of the options, making the problem non-trivial and interesting to study.

\vspace{0.25cm}
\textbf{Our Contribution}. 
We give a probabilistic proof that a policy at distance at most $\floor*{(t - 1) / 2}$ from issue-wise majority which is not opposed by more people than supported always exists. This bound is tight, as shown in \cite{robin_price_of_majority}.
To achieve this, we devise two thought experiments:~one whose expected value is straightforward to compute, but is not immediately instrumental in proving our assertion, and one whose expected value can look daunting to compute, but a non-negative expectation would easily imply the conclusion. By observing voter-wise symmetries of the policy space, we show that the two experiments actually have the same expectation, implying the result. Moreover, unless for all issues the electorate is evenly split between the two options, we can in fact guarantee that support of the proposal strictly exceeds opposition.

Subsequently, we show that a policy satisfying our guarantees can be computed deterministically in polynomial time, by derandomizing our argument. Furthermore, for the odd $t$ case, aside from the pathological case above, an application of Markov's inequality shows that a simpler and more efficient algorithm achieves expected polynomial time:~proceed in rounds, at each round sampling and checking policies $(p_k)_{t / 2 < k \leq t}$ such that $p_k$ agrees with issue-wise majority in exactly $k$ places and is sampled uniformly at random among proposals with this property.

Additionally, we consider the question of determining the minimum distance to issue-wise majority that a policy surviving a vote can achieve. We show that it is NP-hard to decide even whether distance at most $\floor*{(t - 1) / 2} - 1$ is possible, even when both the number of voters and issues are odd. We also investigate the problem from a parameterized perspective, showing that it is tractable with respect to three natural parameters.

\textbf{Appendix}. After the main text, we provide an additional section showing that the space of policies surviving a vote can be, in a certain sense, highly disconnected.

\subsection{Related Work}

Alon et al.~\cite{noga_bredereck_rolf} consider the problem of finding a policy supported by a majority of the voters from the perspective of parameterized complexity. However, in their case a voter supports a policy if and only if they approve of a majority of the issues approved in the policy.
Our objective of minimizing distance to issue-wise majority is also not part of their formulation. Their model of which proposals are supported by a voter is a special case of ``threshold functions'', introduced and first studied computationally by Fishburn and Pekeč \cite{threshold_functions, threshold_functions2}. The search for compromise outcomes when issue-wise majority is defeated is studied by Laffond and Lainé \cite{paradoxes_explained}, where they introduce the concepts of majoritarian and approval compromises. The latter can be seen as a dual to our optimization objective:~instead of optimizing for the number of agreements with issue-wise majority, they optimize for the number of voters supporting the outcome. Elkind et al.~\cite{complexity_deliberative_coalition} perform an algorithmic study of coalition formation in a model similar to ours:~voters support policies which are closer to their vector than to a given ``status quo'' policy.

Voting in combinatorial domains, where the set of admissible outcomes is a subset of some cartesian product, is surveyed from a computational perspective by Lang and Xia \cite{combinatorial_domains}. In our case, this domain is the hypercube $\{0, 1\}^t.$ When issues are symmetric and independent, as in our case, the setup is known as \emph{multiple referenda}. The computational study of multiple referenda so far has predominantly concerned the lobbying problem, where a self-interested third-party wants to manipulate the outcome by changing ballots \cite{earlier_lobbying_in_referenda, lobbying_in_referenda, probabilistic_lobbying}. In \cite{agenda_control_in_referenda}, Conitzer et al.~consider a setup where issues are voted on sequentially, one at a time, and show that it is computationally demanding for a chairperson to influence the outcome by selecting the order the issues are presented in.

Judgement aggregation \cite{judgment_agg} is the generalization of multiple referenda to non-independent issues, and its computational study has been a recently active area of research. Slater's rule is one of the well-known aggregators considered in judgement aggregation, bearing a similar flavour to the topic of our work, in that it computes the closest to issue-wise majority logically-consistent proposal. However, for independent issues, it degenerates to just issue-wise majority, while in our case we additionally require that the selected proposal gets more support than opposition, which is a global constraint incompatible with Slater-like rules. Multiwinner elections are also often studied for approval ballots, leading to a number of prominent computational results \cite{survey_approval}, but they differ from multiple referenda fundamentally in that each issue is equivalent to a corresponding ``negated issue'', while there is no such concept as negating a candidate.

Anscombe's paradox has a number of different interpretations; e.g., read Figure \ref{fig:anscombe_small} as follows:~a government consists of three seats (columns), and two parties (0 and 1) compete. Each voter (row) has preferences for each seat, denoting whose party's nominee they prefer. In a direct democracy, voters vote on each seat, party 1's candidates winning all seats. However, when voting for a single-party government (representative democracy), more voters prefer party 0 over party 1, so party 0 wins all seats. This discrepancy between the outcomes of the two forms of democracy is illustrated through Anscombe's and other related \emph{compound majority paradoxes} \cite{nurmi_compound_paradoxes}.

A large body of literature sets to understand the conditions under which Anscombe's paradox and its generalization, the Ostrogorski Paradox \cite{ostrogorski}, occur and how they can be mitigated. Wagner \cite{wagner_1} observes the ``Rule of $3/4$'' (and later \cite{wagner_2} the generalized ``Rule of $1 - \alpha\beta$''), showing that the paradox requires many contested issues. In particular, if the average issue-wise majority margin is at least $3/4$, then Anscombe's paradox does not occur. This is one reason why certain high-stake votes, like constitutional amendments, require a $3/4$-majority to pass.
Similarly, Laffond and Lainé \cite{unanimity_and_anscombe} note that the paradox requires a non-cohesive group of voters. Namely, if the opinions of any two voters differ in less than a fraction of $\sqrt{2} - 1 \approx 41.4\%$ of the issues, then the paradox again does not occur (moreover, decreasing this constant gives higher guarantees on the fraction of voters agreeing with issue-wise majority). 
Deb and Kelsey \cite{deb_kelsey} give necessary and sufficient conditions in terms of combinations of the number of voters and issues that permit Anscombe's paradox:~unless either is small, the paradox is always possible; they also consider generalizations.
Gehrlein and Merlin \cite{probability_paradox} study the probability of the paradox occurring.
Laffond and Lainé \cite{single_switch} define the domain of single-switch preferences, which is the maximal domain avoiding Ostrogorski's paradox under a mild richness assumption. 
In essence, this constraint, similar to \emph{extremal interval} domains for approval voting \cite{elkind_ssc}, states that the issues can be ordered and possibly negated such that each voter approves of either a prefix or a suffix of issues.


Fritsch and Wattenhofer \cite{robin_price_of_majority} show non-constructively that in our setup a proposal that at most half the voters oppose always exists within distance $\floor*{(t - 1) / 2}.$ If there are no abstentions (odd $t$), this also implies that the opposition can not exceed the support, but with abstentions a problematic case would be when all but three voters abstain, the remaining three being split two on one against our proposal. We improve on their result by tackling both parities of $t$. Our main difference, however, is that we provide polynomial-time means of computing such policies.

\section{Preliminaries}
For any non-negative integer $m$, write $[m] = \{1, 2, \ldots, m\}$ and $\B = \{0, 1\}.$ An assembly of $n$ voters, numbered from $1$ to $n$, expresses preferences over $t$ independent binary \emph{issues} (also topics, or motions), numbered from $1$ to $t$, and collectively known as the \emph{agenda}. Each voter $i$'s opinions on the issues are represented as a vector (ballot) $v_i \in \B^t$, where $v_{i, j}$ is $1$ iff voter $i$ is in favor (approves) of motion $j.$ The matrix $\calP = (v_{i, j})_{i \in [n], j \in [t]}$ is called the voter judgements matrix, or the \emph{preference profile/matrix}. Write $d_H : \B^t \times \B^t \to \Z_{\geq 0}$ for the \emph{Hamming} distance $d_H(x, y) = \sum_{j = 1}^{t}|x_j - y_j|.$ A \emph{policy} (also policy proposal, or outcome) is an element $p \in \B^t;$ we write $\overline{p}$ for the \emph{opposite} policy; i.e., $p_j + \overline{p}_j = 1$ for all issues $j$.
Voter $i$ \emph{supports} (approves) $p$ if $d_H(v_i, p) < t/2$, \emph{opposes} (disapproves) $p$ if $d_H(v_i, p) > t/2,$ and is being \emph{indifferent} (abstains) for $p$ if $d_H(v_i, p) = t/2.$ Write $b_{v_i, p} = t - 2d_H(v_i, p)$ for the number of issues in which $v_i$ and $p$ match minus the number of issues in which $v_i$ and $p$ mismatch. Then, $i$ approves $p$ if $b_{v_i, p} > 0$, disapproves $p$ if $b_{v_i, p} < 0$, and is indifferent if $b_{v_i, p} = 0$.

For a preference profile $\calP$ and a proposal $p$ we define the \emph{for-against balance} $b_{p, \calP} = a_{p, \calP} - d_{p, \calP}$ of $p$ under $\calP$ to be the number $a_{p, \calP}$ of voters supporting $p$ minus the number $d_{p, \calP}$ of voters opposing $p$. As such, we say that a policy $p$ wins/ties/loses the vote if $b_{p, \calP}$ is positive/zero/negative, $p$ being a \emph{winning}/\emph{tying}/\emph{losing} proposal. Note that $p$ survives the vote if it is not losing. A proposal $p$ is \emph{unanimously} winning if $b_{p, \calP} = n.$ For use later on, any ballot (i.e., vote or policy) $p \in \B^t$ can be seen as a function $p : [t] \to \B$, or as a set $p \subseteq [t]$ of approved issues; let $|p|$  be the set cardinality and $b_p = |p| - (t - |p|) = 2|p| - t$ the difference between the number of approved and disapproved issues in $p.$

Given a preference profile $\calP$, the \emph{issue-wise majority policy} (short \IWM) is the policy $p \in \B^t$ such that $p_j$ is $1$ if more than $n / 2$ voters are in favour of motion $j$, is $0$ if less than $n / 2$ voters are, and can be either in case of a tie. For our purposes, we arbitrarily choose the value $1$ in case of equality. Note that any profile $\calP$ where \IWM returns $0$ for some issues can be turned into a logically equivalent profile $\calP'$ by flipping ones and zeros for such issues. This does not modify the electoral landscape because we are just logically negating some of the issues, and all issues are independent; voters should naturally answer by negating their original answers. Henceforth, we assume without loss of generality that all profiles that we consider have all-ones as their \IWM outcome. 

We study the problem of finding a non-losing proposal $p$ which minimizes distance to \IWM; i.e., maximizes $|p|.$ More formally, let \MajSup be the decision problem ``Given a preference profile $\calP$ with \IWM being all-ones and a number $k$, does there exist a non-losing proposal $p$ with $|p| \geq k$?'' Similarly, define $k$-\MajSup, where the value of $k$ is a fixed integer function of $t$; e.g., $(\floor*{t / 2} + 2)$-\MajSup. Additionally, we define the analogous problem \UnaSup, where the sought proposal $p$ has to win unanimously; i.e., be supported by everyone; and also $k$-\UnaSup, similarly.

For odd $t$, Fritsch and Wattenhofer \cite{robin_price_of_majority} show that the answer to $(\floor*{t / 2} + 1)$-\MajSup is always yes, which they show to be tight for all $t$:

\begin{lemma} \label{lemma_gadget} Consider profiles $\calP_t$ with $n = 2t - 1$ voters, where $v_i = \{i\}$ for $i \leq t$ and $v_i = [t]$ for $i > t$. Then, a proposal $p$ with $|p| \geq \floor*{t / 2} + 2$ satisfies $b_{p, \calP_t} = -1.$ Profile $\calP_3$ is illustrated in Figure \ref{fig:anscombe_small}.
\end{lemma}
\begin{proof}[Proof Sketch] The last $t - 1$ voters always approve of $p$, while the first $t$ always disapprove of $p$, so $b_{p, \calP_t} = (t - 1) - t = -1$.
\end{proof}

\section{Finding Non-Losing Proposals is Hard}

In this section we show that $(\floor*{t / 2} + 2)$-\MajSup is NP-hard and as a corollary that $k$-\MajSup is NP-hard for $k \geq (\floor*{t / 2} + 2),$ unless $k$ is subpolynomially close to $t.$ We then show that \MajSup is tractable with respect to three natural parameters.

\begin{theorem} \label{th_maj_reduces_to_una} For any function $k \geq \floor*{t / 2} + 2$, there is a polynomial Karp reduction from $k$-\UnaSup to $k$-\MajSup.
\end{theorem}
\begin{proof}
Consider an $n \times t$ instance $\calP$ of $k$-\UnaSup. From $\calP$ we construct in polynomial time a new instance $\calP'$ by vertically concatenating $n$ copies of the gadget $G = \calP_t$ defined in Lemma \ref{lemma_gadget}. Formally, $\calP' = \calP \mid G_1 \mid \ldots \mid G_n$, where $\mid$ denotes vertical concatenation of preference matrices and $(G_i)_{i \in [n]}$ are copies of $G.$ To show correctness, let $p$ be a proposal with $|p| \geq k.$ 
Proposal $p$ is unanimously winning in $\calP$ iff $b_{p, \calP} = n$, or equivalently $b_{p, \calP} \geq n.$ This is equivalent, by adding $n \cdot b_{p, G}$ to both sides, to $b_{p, \calP} + n \cdot b_{p, G} \geq n + n\cdot b_{p, G}$, which is the same as $b_{p, \calP'} \geq 0$ by Lemma \ref{lemma_gadget}. Crucially, observe that \IWM is also all-ones for $\calP'$, since each column of $G$ has more ones than zeros.
\end{proof}

\begin{theorem} \label{th_una_half_is_np_hard} $(\floor*{t / 2} + 2)$-\UnaSup is NP-hard even when restricted to the odd $n$, odd $t$ case or the even $n$, odd $t$ case.
\end{theorem}
\begin{proof} Before proceeding, we make a few observations and notational conveniences. Namely, we see votes $v_i$ as sets $v_i \subseteq [t]$, and we see potential proposals $p$ as vectors in $\{\pm 1\}^t$, with the correspondence $0 \mapsto -1, 1 \mapsto 1$. Under these assumptions, a voter with vote $v$ approves of proposal $p$ if and only if $\sum_{i \in v} p_i - \sum_{i \notin v} p_i > 0$.
Note that possible inequalities induced by a voter, being of the form $\sum_{i \in v} p_i - \sum_{i \notin v} p_i > 0$, correspond bijectively with inequalities of the form $\sum_{i}c_ip_i > 0$, where $c \in \{\pm 1\}^t$.

Armed as such, we proceed by reduction from Independent Set, similarly to \cite{complexity_deliberative_coalition}. Assume we have a graph $G = (V, E)$, with $|V| = n$ vertices and $|E| = m$ edges, and a number $k$. We will show how to construct a profile $\calP$ in polynomial time such that $G$ has an independent set of size at least $k$ if and only if there is a proposal $p$ with $|p| \geq \floor*{t / 2} + 2$. 

For the agenda, we introduce two issues per vertex of $V$, namely $x_1, x_1', x_2, x_2', \ldots, x_n, x_n'$. Moreover, for some $\ell$ to be chosen later, we also introduce issues $a_1, a_1', a_2, a_2', \ldots, a_\ell, a_\ell'$. Finally, we introduce an additional issue $a_0$. For convenience, let $A = \{a_0, a_1, a_1'$, $ a_2, a_2', \ldots, a_\ell, a_\ell'\}.$

Votes to be added as rows of $\calP$ will be presented in the linear inequality notation from above. We add to $\calP$ the following inequalities:

\vspace{0.25cm}
\textbf{Set 1}. We enforce that for all $\alpha \in A$ it holds that $\alpha = 1$. To do this, write $A$ arbitrarily as $A = \{\alpha\} \cup A_0 \cup A_1$, where $|A_0| = |A_1| = \ell$ and add votes corresponding to the following inequalities:
\begin{gather}
    \alpha + \sum_{y \in A_0}y - \sum_{y \in A_1}y + \sum_{i \in [n]}(x_i - x_i') > 0 \label{eq_a_1_1} \\
    \alpha - \sum_{y \in A_0}y + \sum_{y \in A_1}y - \sum_{i \in [n]}(x_i - x_i')  > 0 \label{eq_a_1_2}
\end{gather}
Note that adding together the two inequalities gives $2\alpha > 0$, which happens iff $\alpha = 1$. However, having $\alpha = 1$ for all $\alpha \in A$ is so far not enough to guarantee that all inequalities in \textbf{Set 1} hold.

\vspace{0.25cm}
\textbf{Set 2}. We enforce that for all $i \in [n]$ it holds that $x_i = x_i'$. To do this, we add votes corresponding to the following inequalities:
\begin{gather}
x_i - x_i' + \sum_{j \in [n] \setminus \{i\}}(x_j - x_j') + a_0 + \sum_{j \in [\ell]} (a_j - a_j') > 0 \label{eq_x_i_x_i_prime_1}\\
x_i - x_i' - \sum_{j \in [n] \setminus \{i\}}(x_j - x_j') + a_0 - \sum_{j \in [\ell]} (a_j - a_j') > 0 \label{eq_x_i_x_i_prime_2}
\end{gather}
Note that adding together the two inequalities gives $2(x_i - x_i' + a_0) > 0$, which is equivalent to $x_i - x_i' + a_0 > 0$, which is the same as $x_i - x_i' \geq 0$ since $a_0 = 1$ holds in the presence of \textbf{Set 1}.
By changing the signs of $x_i$ and $x_i'$ in \eqref{eq_x_i_x_i_prime_1} and \eqref{eq_x_i_x_i_prime_2} we get analogous inequalities \eqref{eq_x_i_x_i_prime_3} and \eqref{eq_x_i_x_i_prime_4} enforcing that $x_i' - x_i \geq 0$, which we also add to $\calP$:
\begin{gather}
x_i' - x_i + \sum_{j \in [n] \setminus \{i\}}(x_j - x_j') + a_0 + \sum_{j \in [\ell]} (a_j - a_j') > 0 \label{eq_x_i_x_i_prime_3}\\
x_i' - x_i - \sum_{j \in [n] \setminus \{i\}}(x_j - x_j') + a_0 - \sum_{j \in [\ell]} (a_j - a_j') > 0 \label{eq_x_i_x_i_prime_4}
\end{gather}
Altogether, we get that $x_i = x_i'$. Now, assuming both \textbf{Set 1} and \textbf{Set 2} have been added, we can see that not only they imply that $\alpha = 1, \forall \alpha \in A$ and $x_i = x_i', \forall i \in [n]$, as we have shown, but also the other way around. In particular, substituting $\alpha = 1, \forall \alpha \in A$ and $x_i = x_i', \forall i \in [n]$ into \eqref{eq_a_1_1}--\eqref{eq_x_i_x_i_prime_4} will each time yield $1 > 0$, so \textbf{Sets 1} and \textbf{2} are satisfied if and only if $\alpha = 1, \forall \alpha \in A$ and $x_i = x_i', \forall i \in [n]$.

\vspace{0.25cm}
\textbf{Set 3}. We enforce the independent set constraints. For each edge $(a, b) \in E$ we want to enforce the constraint that $x_a + x_b < 1$, which is equivalent over the integers to $x_a + x_b < 0.5$. This is the same as $2(x_a + x_b) < 1$, or $x_a + x_b + x_a' + x_b' < 1$, equivalently. Multiplying with $-1$, this is the same as $-x_a - x_b - x_a' - x_b' > -1.$ Finally, this is equivalent to $-x_a -x_b - x_a' - x_b' + a_0 > 0$. Introducing 0 terms to ensure variables are used exactly once, we get the equivalent form, which we add to $\calP$:
\begin{gather}
    -x_a - x_b - x_a' - x_b' + a_0 + \nonumber \\ \sum_{j \in [n] \setminus \{a, b\}}(x_j - x_j') + \sum_{j \in [\ell]} (a_j - a_j') > 0 \label{eq_edge}
\end{gather}

\vspace{0.25cm}
\textbf{Set 4}. Finally, let us enforce the constraint that the independent set has size at least $k$. To do this, note that:
\begin{gather}
    \sum_{i \in [n]:x_i = 1} x_i \geq k \iff 2\sum_{i \in [n]:x_i = 1} x_i \geq 2k \nonumber \\
    \iff \sum_{i \in [n]:x_i = 1}x_i + \left(n + \sum_{i \in [n]:x_i = -1}x_i\right) \geq 2k \nonumber \\
    \iff \sum_{i \in [n]}x_i \geq 2k - n \nonumber \\
    \iff \sum_{i \in [n]}(x_i + x_i') + a_0 > 4k - 2n \label{eq_same_value}
\end{gather}
The last line can be written as $\sum_{i \in [n]}(x_i + x_i') + a_0 + 2(n - 2k) > 0$, so, as long as $\ell \geq \lvert n - 2k \rvert$, we can write $2(n - 2k)$ as a linear $\pm 1$ combination of the variables in $A \setminus \{a_0\}$, thus getting an inequality that uses all variables exactly once, which we can then add to $\calP$.

These being said, let us now investigate the (implicit in the definition of $(\floor*{t / 2} + 2)$-\UnaSup) constraint that at least $\floor*{t / 2} + 2$ of the variables in the sought outcome have to take value $1$. In our case, $t = 2n + 2\ell + 1.$ Like before, the first step is to translate from the number of ones to the sum of the variables. Doing so, at least $\floor*{t / 2} + 2$ ones in the outcome is equivalent to:
\begin{gather*}
    \sum_{i \in [n]}{(x_i + x_i')} + a_0 +  \sum_{i \in [\ell]}{(a_i + a_i')} \geq 2(\floor*{t / 2} + 2) - (2n + 2\ell + 1) \\
    \iff \sum_{i \in [n]}{(x_i + x_i')} + a_0 \geq 3 - 2\ell
\end{gather*}
The quantity in the left-hand side is the same as in \eqref{eq_same_value}, and note that $4k - 2n \geq 3 - 2\ell \iff \ell \geq n + 3/2 - 2k$. When this happens, the inequality in \textbf{Set 4} is stronger than the constraint on the total number of ones, so we can ignore the constraint on the number of ones. Therefore, assuming that $k \geq 1$, we can set $\ell = n$ such that both $\ell \geq \lvert n - 2k \rvert$ and $4k - 2n \geq 3 - 2\ell$ hold true.

There is one more thing to take care of:~we need to ensure that in the so-constructed judgements matrix $\calP$ every column has more ones than zeros, and this is not yet true for our construction. To mitigate the issue, note the following for constraint \textbf{Sets 4} and \textbf{1}:

\textbf{Set 1}. In Equations \eqref{eq_a_1_1} and \eqref{eq_a_1_2} every variable appears once positive and once negative, except for $\alpha$ itself, which appears twice positive. This being said, if at the end we are left with a matrix violating the required property in some column corresponding to a variable in $A$, we can add copies of the corresponding inequalities in \textbf{Set 1} polinomially many times until such violations are no longer present, since each added copy only changes the one-zero balance of a single column, increasing it.   

\textbf{Set 4}. Here only variables in $A$ ever appear negative. Therefore, if at some point we have a matrix violating the property for some variable $x_i$ or $x_i'$, then we add an additional copy of \textbf{Set 4}, hence increasing the one-zero balance of variables $x_i$ and $x_i'$ by 1. After polynomially many such additions, we reach a matrix satisfying the property for all columns corresponding to vertices in the graph.

These being said, our reduction can now be summarized as follows:~first, we build a judgements matrix $\calP$ as described, then, we perform the transformation described for \textbf{Set 4} above until all variables of the form $x_i$ and $x_i'$ have positive one-zero balances on their columns, and finally we perform the transformation described for \textbf{Set 1} above until also all columns corresponding to variables in $A$ have positive one-zero balances. Observe that the produced instance has $t = 2n + 2\ell + 1,$ which is odd. To control the parity of the number of voters, we can add one more copy of \textbf{Set 4} in the post-processing stage if needed, as the subsequent transformations using \textbf{Set 1} do not change the parity of the number of voters. 
\end{proof}

Putting together Theorem \ref{th_maj_reduces_to_una} and Theorem \ref{th_una_half_is_np_hard} we get our main result as the following:

\begin{theorem} \label{th_main_np_hardness_result} $(\floor*{t / 2} + 2)$-\MajSup is NP-hard even when restricted to the odd $n$, odd $t$ case or the even $n$, odd $t$ case.\footnote{For the odd $n$, odd $t$ case we actually need to use $\calP' = \calP \mid G_1 \mid \ldots \mid G_{n - 1}$ in the proof of Theorem \ref{th_maj_reduces_to_una}, but the details are analogous.}
\end{theorem}

However, maybe checking for the existence of a policy with more agreements is not as difficult? After all, checking for distance at most any given constant can be done in polynomial time. The answer is mostly negative, which we show in the following more technical form of Theorem \ref{th_una_half_is_np_hard}:

\begin{theorem} \label{th_np_hard_stronger_form} $k$-\UnaSup is NP-hard if $t - k$ is $\Omega(t^\epsilon)$ for some constant $0 < \epsilon \leq 1$. 
\end{theorem}
\begin{proof}
In order to distinguish $k$, which is the size of the sought independent set, from $k$, which is the parametrization of \UnaSup for which we want to reduce to, we denote the latter with $f(t)$. In other words, we are reducing Independent Set to $f(t)$-\UnaSup.

Assume $0 < \epsilon \leq 1$ is such that $t - f(t)$ is $\Omega(t^\epsilon)$. In particular, assume $C > 0$ and $t_0 \geq 1$ are such that for all $t \geq t_0$ we have ${t - f(t) \geq Ct^\epsilon}$. The proof proceeds very similarly to that of Theorem \ref{th_una_half_is_np_hard}, except for choosing $\ell$. First, like before, we need to have $\ell \geq \lvert n - 2k \rvert$, in order to allow constraint \textbf{Set 4} to be written. Second, the (implicit) constraint that at least $f(t)$ variables have to be set to 1 is now written as follows, where recall that $t = 2n + 2\ell + 1$:

\begin{gather*}
    \sum_{i \in [n]}{(x_i + x_i')} + a_0 +  \sum_{i \in [\ell]}{(a_i + a_i')} \geq 2f(t) - (2n + 2\ell + 1) \\
    \iff \sum_{i \in [n]}{(x_i + x_i')} + a_0 \geq 2f(2n + 2\ell + 1) - (2n + 4\ell + 1)
\end{gather*}

As before, we want to have this constraint dominated by the one for the size of the independent set, so we require that $4k - 2n \geq 2f(2n + 2\ell + 1) - (2n + 4\ell + 1)$, which is equivalent to $4k \geq 2f(2n + 2\ell + 1) - (4\ell + 1).$ We now want to show that there is an $\ell \geq \lvert n - 2k \rvert$ satisfying this constraint of value bounded by a polynomial in $n$ (since otherwise our reduction would no longer produce a polynomially-sized matrix). To do so, we will show that such an $\ell$ exists satisfying the stronger condition $2f(2n + 2\ell + 1) \leq 4\ell + 1$, and, in fact $f(2n + 2\ell + 1) \leq 2\ell.$ Since for $2n + 2\ell + 1 \geq t_0$ we have that $f(2n + 2\ell + 1) \leq 2n + 2\ell + 1 - C(2n + 2\ell + 1)^\epsilon$, it suffices to find $\ell$ such that $2n + 2\ell + 1 - C(2n + 2\ell + 1)^\epsilon \leq 2\ell.$ The last condition is equivalent to $2n + 1 \leq C(2n + 2\ell + 1)^\epsilon,$ and in turn to
\[
\left(\frac{2n + 1}{C}\right)^{1/\epsilon} \leq 2n + 2\ell + 1 \iff \ell \geq \frac{1}{2}\left(\left(\frac{2n + 1}{C}\right)^{1/\epsilon} - 1\right) - n
\]
Denote the value of the latter right-hand side with $x$. Since $x$ is polynomial in $n$ for any fixed constants $C > 0$ and $0 < \epsilon \leq 1$, it follows that we can pick $\ell = \ceil*{\max\{t_0, \lvert n - 2k \rvert, x\}}$ to complete the proof.
\end{proof}

\begin{corollary}
$k$-\MajSup is NP-hard if $k \geq \floor*{t / 2} + 2$ and $t - k$ is $\Omega(t^\epsilon)$ for some constant $0 < \epsilon \leq 1$.
For instance, deciding whether a non-losing outcome agreeing in at least $99\%$ of all issues, or in at least $t - \sqrt{t}$ issues is NP-hard.
\end{corollary}

We now turn our attention to parameterized complexity.
\begin{theorem} \MajSup is FPT with respect to any of the parameters $n$, $t$ and $h = \max_{i, j}{d_H(v_i, v_j)}.$ Moreover, \UnaSup is FPT with respect to $n$ and $t.$
\end{theorem}

\begin{proof} For parameter $t$, a straightforward enumeration of all proposals $p$ followed by counting the number of voters approving and disapproving $p$ already achieves complexity $O(nt2^t)$ for both problems, proving the claim. For parameter $h$, the situation is not very different for $k$-\MajSup:~in \cite{unanimity_and_anscombe} it is shown that, for $t$ above a certain threshold, if $h < (\sqrt{2} - 1)t,$ then issue-wise majority does not lose, fact which can be easily checked for before proceeding further. Otherwise, we know that $h \geq (\sqrt{2} - 1)t$, in which case $t \leq (1 + \sqrt{2})h.$ Hence, exhaustive search in this case runs in time $O(nt2^t) \subseteq O(nh2^{(1 + \sqrt{2})h}) \subseteq O(n5.34^h).$ Finally, the proofs of tractability with respect to $n$ are more involved, and are presented next.

We begin with $k$-\UnaSup. The proof proceeds by formulating the problem as an integer linear program, similarly to \cite{complexity_deliberative_coalition}. The main insight here is that there can be at most $2^n$ distinct columns in the preference matrix. For every possible column $c \in \B^n$, let $t_c$ denote the number of columns of type $c$; i.e.~columns identical to $c$; present in the matrix. Since identical columns are, in all practical regards, interchangeable, introduce variables $(x_c)_{c \in \B^n},$ where $x_c$ denotes for how many of the $t_c$ columns of type $c$ the sought proposal will have a one at the corresponding positions. Note that then $t_c - x_c$ columns will have a zero at the corresponding positions. With the following inequality, added for every $c \in \B^n$, we enforce consistency requirements:
\begin{equation} \label{cs_1}
    0 \leq x_c \leq t_c
\end{equation}
The following inequality is our ``optimization objective'':
\begin{equation} \label{cs_2}
\sum_{c \in \B^n}x_c \geq k
\end{equation}
The following inequality, added for every voter $i \in [n],$ asserts that $i$ does not disapprove of the proposal described by $x$:
\begin{equation*}
\sum_{\substack{c \in \B^n \\ c_i = 1}}x_c + \sum_{\substack{c \in \B^n \\ c_i = 0}}(t_c - x_c) > \frac{t}{2}
\end{equation*}
\noindent which, since the $t_c$ terms sum to $t - |v_i|$, is equivalent to
\begin{equation*}
\sum_{\substack{c \in \B^n \\ c_i = 1}}x_c - \sum_{\substack{c \in \B^n \\ c_i = 0}}x_c > |v_i| - \frac{t}{2}
\end{equation*}
\noindent which can be rewritten in a more standard form as:
\begin{equation} \label{cs_3}
\sum_{\substack{c \in \B^n \\ c_i = 1}}x_c - \sum_{\substack{c \in \B^n \\ c_i = 0}}x_c \geq |v_i| - \floor*{\frac{t - 1}{2}}
\end{equation}
Together, the constraint sets described by \eqref{cs_1}, \eqref{cs_2}, and \eqref{cs_3} make up our ILP. The number of variables is $2^n$ and the number of constraints is $2^{n + 1} + 1 + n,$ meaning that the associated system matrix is of shape $(2^{n + 1} + n + 1) \times (2^n + 1)$, the total number of cells hence being $O(4^n).$ The coefficients in the matrix are bounded in absolute value by $t$, meaning that the ILP instance can be stored in $O(4^n\log t)$ bits. Lenstra’s Algorithm (and subsequent improvements of it) solve an ILP in time exponential in the number of variables, but linear in the number of bits required to represent the instance matrix, so our result follows.

For $k$-\MajSup, we keep the same ILP-based proof idea, but the details need some refining. We keep the same variables and all inequalities arising from \eqref{cs_1} and \eqref{cs_2}, but \eqref{cs_3} needs to be replaced by inequalities signifying that the proposal described by $x$ does not lose. To model this, introduce $3n$ additional binary variables $(a_i)_{i \in [n]}, (d_i)_{i \in [n]}, (f_i)_{i \in [n]}$, with the meaning that $a_i$ is 1 if and only if voter $i$ approves of the proposal, $d_i$ is 1 if and only if voter $i$ disapproves of the proposal, and $f_i$ is 1 if and only if voter $i$ is indifferent for the proposal. The following inequalities, added for every voter $i \in [n],$ enforce consistency requirements:
\begin{gather}
    a_i, d_i, f_i \geq 0 \label{cs_4} \\
    a_i + d_i + f_i = 1 \label{cs_5}
\end{gather}
To enforce that every voter $i \in [n]$ with $a_i = 1$ actually approves of the proposal, we use the following inequality:
%
\begin{gather}
\sum_{\substack{c \in \B^n \\ c_i = 1}}x_c - \sum_{\substack{c \in \B^n \\ c_i = 0}}x_c \geq |v_i| - \floor*{\frac{t - 1}{2}} - \beta(1 - a_i) \label{cs_6}
\end{gather}
where $\beta = 10t + 10$ was chosen as a large-enough constant linear in $t$. Note that for $a_i = 1$ the inequality above becomes \eqref{cs_3}, while for $a_i = 0$ it is always satisfied because the left-hand side is always between $-t$ and $t.$ Similarly, to enforce that every voter $i \in [n]$ with $d_i = 1$ actually disapproves of the proposal, we use the following similar inequality:
\begin{gather}
\sum_{\substack{c \in \B^n \\ c_i = 1}}x_c - \sum_{\substack{c \in \B^n \\ c_i = 0}}x_c \leq |v_i| - \floor*{\frac{t - 1}{2}} - 2 + \beta(1 - d_i) \label{cs_7}
\end{gather}
Finally, to enforce that every voter $i \in [n]$ with $f_i = 1$ is actually indifferent for the proposal, we use the following inequalities:
\begin{gather}
\sum_{\substack{c \in \B^n \\ c_i = 1}}x_c - \sum_{\substack{c \in \B^n \\ c_i = 0}}x_c \leq |v_i| - \floor*{\frac{t - 1}{2}} - 1 + \beta(1 - f_i) \label{cs_8}\\
\sum_{\substack{c \in \B^n \\ c_i = 1}}x_c - \sum_{\substack{c \in \B^n \\ c_i = 0}}x_c \geq |v_i| - \floor*{\frac{t - 1}{2}} - 1 - \beta(1 - f_i) \label{cs_9}
\end{gather}
Together, the constraint sets described by \eqref{cs_1}, \eqref{cs_2} and \eqref{cs_4}--\eqref{cs_9} make up our ILP. The associated system matrix is of shape $(2^{n + 1} + 9n + 1) \times (2^n + 3n + 1).$ The number of cells is $O(4^n),$ each taking $O(\log t)$ bits to store, so the ILP instance can be stored in $O(4^n\log t)$ bits. Using Lenstra's Algorithm, the conclusion again follows.
\end{proof}

\section[The Case k = t / 2 + 1]{The Case $k = \floor*{t / 2} + 1$}

We have seen that $(\floor*{t / 2} + 2)$-\MajSup is NP-hard.
Fritsch and Wattenhofer \cite{robin_price_of_majority} gave a nonconstructive proof that if $t$ is odd then there exists a non-losing proposal $p$ with $|p| \geq \floor*{t / 2} + 1.$ Their proof hinges on an algebraic combinatorial identity, which the authors found difficult to interpret intuitively. Moreover, they left finding such a proposal in polynomial time open. The purpose of this section is threefold:~first, by uncovering voter-wise structure preserving symmetries over the space of proposals, we give a new probabilistic argument for their result, extending also to the even $t$ case; then, we show how this leads to a simple randomized algorithm running in expected polynomial time assuming odd $t$ and $\Delta > 0$ (see notation below); finally, using derandomization, we give a deterministic polynomial-time algorithm for finding a non-losing proposal $p$ with $|p| \geq \floor*{t / 2} + 1,$ this time for general $t$ and $\Delta.$ We note that, while the deterministic algorithm can handle general instances, the randomized algorithm is more efficient and easier to describe and implement.

Throughout the section, we work with a fixed $n \times t$ preference profile $\calP.$ We write $\Delta$ for the total number of ones minus zeros in matrix $\calP.$ Since we assumed that issue-wise majority is the all-ones vector, it follows that $\Delta \geq 0$ is guaranteed.

\subsection{Structure-Preserving Symmetries and a Probabilistic Proof}
In this section for each voter we construct two structure-preserving bijective correspondences between proposals. We use these to derive a third correspondence with the property that the signs of $b_{p}$ and $b_{v, p}$ are preserved as long as they are non-zero, fact which will be instrumental in our proof. We believe these observed symmetries to be of independent interest. Afterwards, we define two probabilistic thought experiments:~one whose expectation is easy to compute, and can be seen to be non-negative, but alone does not mean much, and one whose expectation might appear difficult to compute, but a non-negative value would easily imply our conclusion. Using the third correspondence above we then deduce that the two expectations are equal, implying the result of \cite{robin_price_of_majority} for arbitrary parity of $t.$

To begin, consider some voter whose vote is $v$ with $|v| = \ell.$ Let $B^* \subseteq \B^t$ be the set of proposals $p$ such that $b_p \neq 0$ and $b_{v, p} \neq 0.$
We classify proposals in $B^*$ into four distinct categories, called $T_{00}, T_{01}, T_{10}, T_{11}$. In particular, the first bit of the subscript is 0/1 depending on whether $b_p$ is negative/positive, while the second bit is 0/1 depending on whether $b_{v, p}$ is negative/positive. We now define two bijective maps $f_{v, 1}, f_{v, 0} : B^* \to B^*.$ We abbreviate $f_{v, 1}$ and $f_{v, 0}$ to simply $f_1$ and $f_0$ when $v$ is clear from context. Without loss of generality, assume $v$ has ones in its first $\ell$ issues and zeros in the other $t - \ell$ issues. Consequently, any proposal $p \in B^*$ can be written as $p = p_1p_0$, where $p_1 \in \B^{\ell}$ and $p_0 \in \B^{t - \ell}$. For convenience, we also write $v = v_1v_0,$ where $v_1 = 1^\ell$ and $v_0 = 0^{t - \ell}.$ The two bijective maps are then given by $f_1(p_1p_0) = \overline{p_1}p_0$ and $f_0(p_1p_0) = p_1\overline{p_0}$.

\begin{lemma} \label{lemma_f_0} $f_0$ maps proposals of type $T_{ij}$ to proposals of type $T_{ji}$, for $i, j \in \{0, 1\}$. Moreover, for any $p \in B^*$ we have $b_{v, f_0(p)} = b_p.$
\end{lemma}
\begin{proof} Consider a proposal $p = p_1p_0 \in B^*.$ Then, $b_{f_0(p)} = b_{p_1\overline{p_0}} = b_{p_1} + b_{\overline{p_0}} = b_{v_1, p_1} + b_{v_0, p_0} = b_{v, p}.$ Similarly, $b_{v, f_0(p)} = b_{v_1v_0, p_1\overline{p_0}} = b_{v_1, p_1} + b_{v_0, \overline{p_0}} = b_{p_1} + b_{p_0} = b_{p_1p_0} = b_p.$ Therefore, $f_0$ swapped the two quantities $b_p$ and $b_{v, p}$, meaning that if $p$ is of type $T_{ij}$, then $f_0(p)$ is of type $T_{ji}.$
\end{proof}

\begin{lemma} \label{lemma_f_1} $f_1$ maps proposals of type $T_{ij}$ to type $T_{(1 - j)(1 - i)}$ proposals, for $i, j \in \{0, 1\}$. Moreover, for any $p$ we have $b_{v, f_1(p)} = -b_p.$
\end{lemma}
\begin{proof} 
Consider a proposal $p = p_1p_0 \in B^*.$ Then, $b_{f_1(p)} = b_{\overline{p_1}p_0} = b_{\overline{p_1}} + b_{p_0} = -b_{v_1, p_1} - b_{v_0, p_0} = -b_{v, p}.$ Similarly, $b_{v, f_1(p)} = b_{v_1v_0, \overline{p_1}p_0} = b_{v_1, \overline{p_1}} + b_{v_0, p_0} = -b_{p_1} - b_{p_0} = -b_{p_1p_0} = -b_p.$ Therefore, $f_1$ swapped and negated the two quantities $b_p$ and $b_{v, p},$ meaning that if $p$ is of type $T_{ij}$, then $f_1(p)$ is of type $T_{(1 - j)(1 - i)}.$
\end{proof}

\begin{figure}[t]
    \centering
    \begin{subfigure}{0.43\linewidth}
    \centering
    \includegraphics[width=0.95\linewidth]{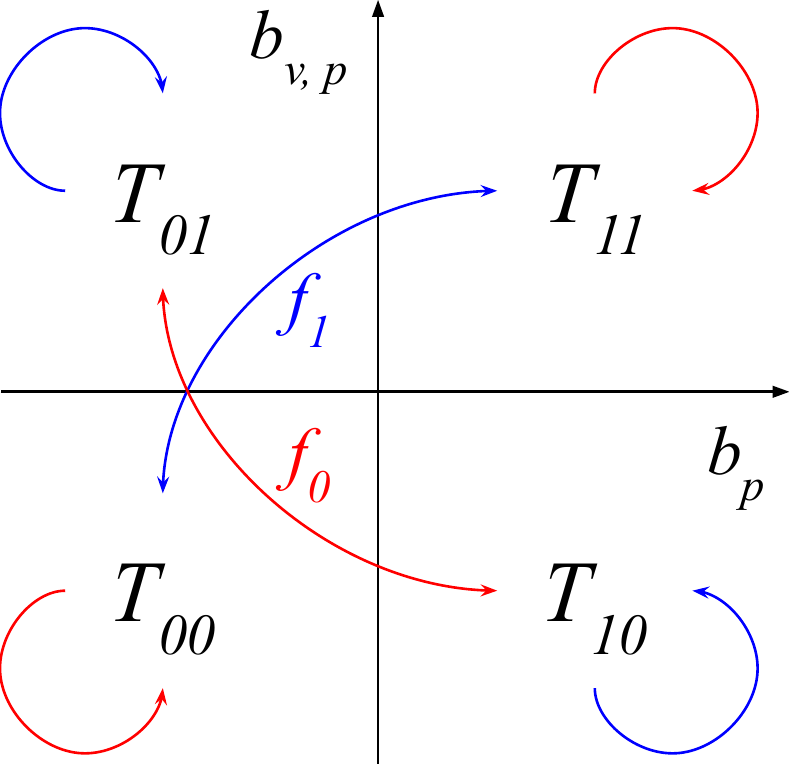}
    \caption{Bijections $f_{v, 0}$ and $f_{v, 1}.$}
    \label{fig:bijection_f_0_1}
    \end{subfigure}
    \begin{subfigure}{0.43\linewidth}
    \centering
    \includegraphics[width=0.95\linewidth]{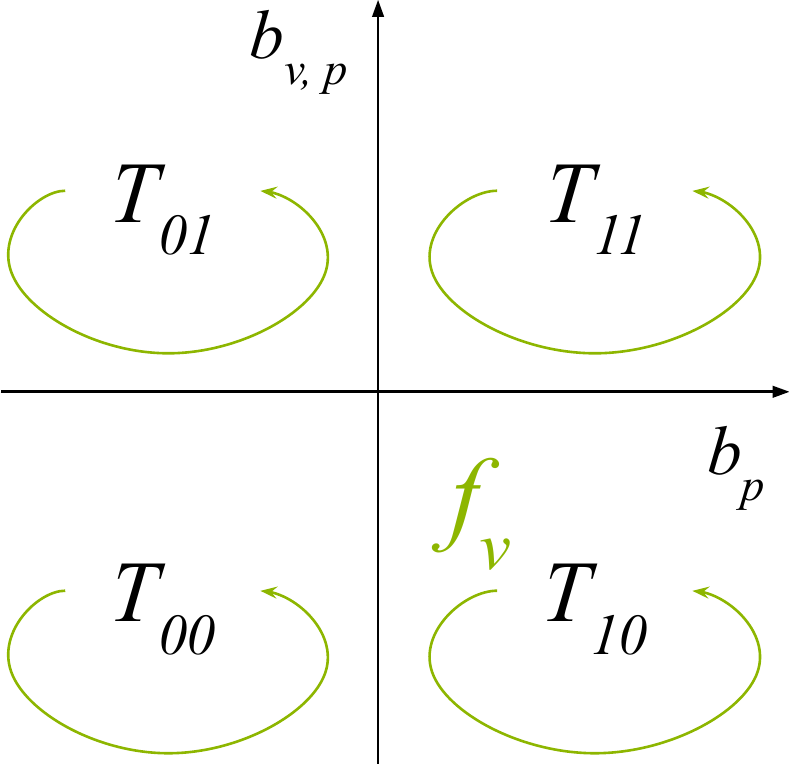}
    \caption{Bijection $f_v.$}
    \label{fig:bijection_f_v}
    \end{subfigure}
    \caption{Bijections $f_{v, 0}, f_{v, 1}$ and $f_v$ corresponding to vote $v$.}
    \label{fig:bijection}
    \Description{The bijective maps used in our proof.}
    \label{fig:bijections}
\end{figure}

Lemma \ref{lemma_f_0} and Lemma \ref{lemma_f_1} are illustrated in Figure \ref{fig:bijection_f_0_1}.

Using $f_{v, 0}$ and $f_{v, 1}$ we can now define a ``mixed'' bijective map $f_v : B^* \to B^*$, as follows:~let $p \in B^*$ be arbitrary, if $p$ is of types $T_{00}$ or $T_{11}$, then $f_v$ maps $p \mapsto f_{v, 0}(p)$, otherwise $f_v$ maps $p \mapsto f_{v, 1}(p)$. Note that the three bijective maps are self-inverse. The new map $f_v$ inherits the properties of the other maps from Lemmas \ref{lemma_f_0} and \ref{lemma_f_1}:

\begin{corollary} \label{coro_f} $f_v$ maps proposals of type $T_{ij}$ to proposals of type $T_{ij}$. Moreover, for any proposal $p$ of type $T_{00}$ or $T_{11}$ we have $b_{v, f_v(p)} = b_p$, and for any proposal $p$ of type $T_{01}$ or $T_{10}$ we have $b_{v, f_v(p)} = -b_p$.
\end{corollary}
\begin{proof}
First, consider a proposal $p$ of type $T_{ij}$ for $ij \in \{00, 11\}.$ By definition, $f_v$ maps $p \mapsto f_{v, 0}(p)$. By Lemma \ref{lemma_f_0} we have that $f_{v, 0}(p)$ is of type $T_{ji}$, which is the same as $T_{ij}$ for $ij \in \{00, 11\}$, and that $b_{v, f_0(p)} = b_p.$ Second, consider a proposal $p$ of type $T_{ij}$ for $ij \in \{01, 10\}.$ By definition, $f_v$ maps $p \mapsto f_{v, 1}(p)$. By Lemma \ref{lemma_f_1} we have that $f_{v, 1}(p)$ is of type $T_{(1-j)(1-i)}$, which is the same as $T_{ij}$ for $ij \in \{01, 10\}$, and that $ b_{v, f_1(p)} = -b_p,$ finishing the proof. 
\end{proof}

Corollary \ref{coro_f} is illustrated in Figure \ref{fig:bijection_f_v}. Next, we introduce the first thought experiment, which alone would not be very helpful, but its expectation is relatively straightforward to compute. For brevity, introduce the notation $B_m = \{p \in \B^t : |p| > t / 2\}$, whose size is $|B_m| = 2^{t - 1}$ for odd $t$ and $|B_m| = 2^{t - 1} - \frac{1}{2}\binom{t}{t / 2}$ for even $t$.

\paragraph{\textbf{Thought Experiment TE1}} We sample a proposal $p \in B_m$ uniformly at random and start with a global counter $X = 0$; each voter $i$ then looks at their own vote $v_i$ and adds $1$ to $X$ for each position in which $v_i$ and $p$ agree and subtracts $1$ from $X$ for each position where they disagree. Overall, voter $i$ adds $b_{v_i, p}$ to $X$. We are interested in the expected value $\E[X]$.

\begin{theorem} $\E[X] = \Delta |B_m|^{-1}\binom{t - 1}{t / 2}$, where $\Delta$ denotes the number of ones minus the number of zeros in the preference profile $\calP$.
\end{theorem}
\begin{proof} Write $X = \sum_{i = 1}^n\sum_{j = 1}^{t}X_{i,j}$, where $X_{i, j} = 1$ if $v_{i, j} = p_j$, and $-1$ otherwise. By linearity of expectation 
\[
\E[X] = \E\left[\sum_{i = 1}^n\sum_{j = 1}^{t}X_{i,j}\right] = \sum_{i = 1}^n\sum_{j = 1}^{t}\E\left[X_{i,j}\right]
\]

\noindent so we only need to compute $\E\left[X_{i,j}\right]$. For the case $v_{i, j} = 1$, we have that $\E\left[X_{i,j}\right] = \P(p_j = 1) - \P(p_j = 0)$, while for the case $v_{i, j} = 0$, we have that $\E\left[X_{i,j}\right] = \P(p_j = 0) - \P(p_j = 1)$. Therefore, we get that $\E[X] = \Delta(\P(p_j = 1) - \P\left(p_j = 0)\right)$. Because $p \in B_m$ is sampled uniformly, it follows that $\P(p_j = 1) - \P(p_j = 0) = |B_m|^{-1}(N_1 - N_0),$ where $(N_k)_{k \in \{0, 1\}}$ is the total number of proposals $p \in B_m$ such that $p_j = k.$ Hence, it remains to compute $N_1 - N_0$.

By symmetry, $(N_k)_{k \in \{0, 1\}}$ do not depend on $j$, so assume without loss of generality that $j = 1$. Consider the bijective map $q : \B^t \to \B^t$ flipping the first entry in the proposal. Consider a proposal $p \in B_m$ such that $p_1 = 1$. If $|p| > \floor*{t / 2} + 1$, then $|q(p)| \geq \floor*{t / 2} + 1$ and $q(p)_1 = 0$ hold, so $p$ and $q(p)$ cancel each other out in $N_1 - N_0.$ This leaves us with counting in $N_1$ those proposals $p \in \B^t$ with $|p| = \floor*{t / 2} + 1$ and $p_1 = 1$, which are the only ones not accounted for. This is the number of ways to choose $\floor*{t / 2}$ entries from $t - 1$ available slots, as $p_1$ is set to $1$, equalling $\binom{t - 1}{\floor*{t / 2}}.$
\end{proof}

Since $\Delta \geq 0$, it follows that $\E[X] \geq 0$, but this is not very useful on its own, as it only implies that there is a proposal $p \in B_m$ with $\sum_{i = 1}^n b_{v_i, p} \geq 0,$ which is already true for $p$ being the all-ones vector, which might lose against all-zeros; i.e., Anscombe's paradox. We now define a slightly less natural thought experiment, computing a number $Y$. Afterwards, we will use our voter-wise bijective maps $f_{v_i}$ to conclude that $\E[Y] = \E[X]$, from which our main result will follow.

\paragraph{\textbf{Thought Experiment TE2}} We sample a proposal $p \in B_m$ uniformly at random and start with a global counter $Y = 0$. Each voter $i$ then compares $p$ with their own vote $v_i.$ If $i$ approves of $p$, they add $b_p$ to $Y$; if $i$ disapproves of $p$, they subtract $b_p$ from $Y$; and if $i$ is indifferent for $p$, then they leave $Y$ unchanged. We are interested in the expected value $\E[Y]$. \\

For a proposal $p \in \B^t$, recall that $b_{p, \calP} = a_{p, \calP} - d_{p, \calP}$ is the number of voters approving $p$ minus the number of voters disapproving $p$ in preference profile $\calP$. Knowing this, it is useful to note that $Y = b_p \cdot b_{p, \calP} = (2|p| - t) \cdot b_{p, \calP}$. The following surprising connection constitutes the key insight in our argument.

\begin{theorem} \label{th_exp_y} $\E[X] = \E[Y] = \Delta |B_m|^{-1}\binom{t - 1}{t / 2}$
\end{theorem}

To prove this, write $Y = \sum_{i \in [n]}Y_i$, where $Y_i$ is $b_p$ if $i$ approves of $p$, is $-b_p$ if $i$ disapproves of $p$, and is $0$ otherwise. Recall that $X = \sum_{i \in [n]}{b_{v_i, p}}$. By linearity of expectation, it suffices to show that $\E[b_{v_i, p}] = \E[Y_i],$ which we do in the following.

\begin{lemma} For any voter $i \in [n]$, we have $\E[b_{v_i, p}] = \E[Y_i].$
\end{lemma}
\begin{proof} Let $B_{+}$ and $B_{-}$ be the sets of proposals in $B_m$ that $i$ approves and disapproves of, respectively. For brevity, write $v = v_i$. Then, we can write:
\begin{gather*}
\E[b_{v, p}] = |B_m|^{-1}\sum_{p \in B_m}{b_{v, p}} = |B_m|^{-1}\left(\sum_{p \in B_+}{b_{v, p}} + \sum_{p \in B_-}{b_{v, p}}\right) \\
= |B_m|^{-1}\left(\sum_{p \in B_+}{b_{f_{v}(p)}} - \sum_{p \in B_-}{b_{f_{v}(p)}}\right)
\end{gather*}
For the last equality, we used Corollary \ref{coro_f} and the fact that $f_v$ is self-inverse. Moreover, from the same we know that $p \in B_s$ iff $f_{v}(p) \in B_s$, for any $s \in \{+, -\}$. Since $f_{v}$ is a bijection, we can therefore make the change of variable $p \mapsto f_{v}(p)$ in both sums to get:
\begin{gather*}
= |B_m|^{-1}\left(\sum_{p \in B_+}{b_p} - \sum_{p \in B_-}{b_p}\right) = \E[Y_i]
\end{gather*}
The fact that the last sum equals $\E[Y_i]$ followed immediately from the definition of $Y_i$.
\end{proof}

From this, out main result follows:

\begin{theorem} \label{th_maj_sup_outcome_exists} Assuming $\Delta \geq 0$, there exists a non-losing proposal $p \in B_m.$ Assuming $\Delta > 0$, there exists a winning proposal $p \in B_m.$
\end{theorem}
\begin{proof} From Theorem \ref{th_exp_y} we get that $\E[Y] = \Delta|B_m|^{-1}\binom{t - 1}{\floor*{t/2}}$, which is $\geq 0$ when $\Delta \geq 0$ and $> 0$ when $\Delta > 0$. As a result, given that $Y = b_p \cdot b_{p, \calP}$, where $b_p = 2|p| - t > 0$ for $p \in B_m$, it follows there is a proposal $p \in B_m$ with $b_{p, \calP} \geq 0$ for $\Delta \geq 0$ and $b_{p, \calP} > 0$ for $\Delta > 0,$ completing the proof.
\end{proof} 

\subsection{Polynomial Computation of Winning Proposals}

From Theorem \ref{th_maj_sup_outcome_exists} we know that a non-losing (winning) proposal $p \in B_m$ always exists, but a polynomial algorithm for finding it is not guaranteed. In this section, we provide two such algorithms:~a simple and relatively efficient randomized algorithm with expected polynomial runtime for odd $t$ and $\Delta > 0$, as well as a more intricate deterministic polynomial-time algorithm for the general case. 

We begin with a lemma which will be useful for both algorithms. Introduce the notation $\P_k[\cdot] = \P[\cdot \text{ given }|p| = k]$ and similarly $\E_k.$

\begin{lemma} \label{lemma_exists_k} There exists $k \geq \floor*{t / 2} + 1$ s.t.~$\E_k[b_{p, \calP}] \geq \frac{\Delta}{t|B_m|}\binom{t - 1}{\floor*{t / 2}}.$
\end{lemma}
\begin{proof} Write the expectation of $Y$ as follows:
\begin{gather*}
\E[Y] = \sum_{k = \floor*{t / 2} + 1}^t \E_k[Y] \P(|p| = k)
\end{gather*}
From Theorem \ref{th_exp_y}, we know that $\E[Y] = \Delta |B_m|^{-1}\binom{t - 1}{\floor*{t / 2}}$. Since the sum of the $\P(|p| = k)$ coefficients is $1$, this means that there is some number $k \geq \floor*{t / 2} + 1$ such that $\E_k[Y] \geq \Delta |B_m|^{-1}\binom{t - 1}{\floor*{t / 2}}$. Since $Y = b_p \cdot b_{p, \calP} = (2|p| - t)\cdot b_{p, \calP}$, this means that $\E_k[Y] = (2k - t)\E_k[b_{p, \calP}]$, from which
\begin{gather*}    
\E_k[b_{p, \calP}] = \frac{\Delta}{(2k - t)|B_m|}\binom{t - 1}{\floor*{t / 2}} \geq \frac{\Delta}{t|B_m|}\binom{t - 1}{\floor*{t / 2}} \qedhere 
\end{gather*}
\end{proof}

\paragraph{\textbf{Randomized Algorithm}} Here we assume that $t$ is odd and $\Delta > 0$. Since $t$ is odd, note that $|B_m|^{-1} = 2^{1 - t},$ and, moreover, that $a_{p, \calP} = n - d_{p, \calP},$ from which $b_{p, \calP} = 2a_{p, \calP} - n.$ Additionally, note that in this case $\P_k(b_{p, \calP} > 0) = \P_k(a_{p, \calP} > n / 2).$ From this, employing Markov's inequality and bounding the central binomial coefficient using a Stirling-type result leads to the following:

\begin{lemma} \label{lemma_exists_k_implies_a} Assume $k$ is such that $\E_k[b_{p, \calP}] \geq \frac{\Delta 2^{1 - t}}{t}\binom{t - 1}{\floor*{t / 2}},$ then $\P_k(b_{p, \calP} > 0) \geq \sqrt{\frac{2}{\pi}}\frac{\Delta}{nt^{3/2}}.$
\end{lemma}
\begin{proof}
Recall that $b_{p, \calP} = 2a_{p, \calP} - n$, so we can write
\begin{gather*}    
\E_k[b_{p, \calP}] = 2\E_k[a_{p, \calP}] - n \iff \E_k[a_{p, \calP}] = \frac{1}{2}\E_k[b_{p, \calP}] + \frac{n}{2}
\end{gather*}
meaning that
\begin{gather*}
\E_k[a_{p, \calP}] \geq \frac{n}{2} + \frac{\Delta 2^{-t}}{t}\binom{t - 1}{\floor*{t / 2}}
\end{gather*}
Now, note that $\P_k(b_{p, \calP} > 0) = \P_k(a_{p, \calP} > n/2) = 1 - \P_k(a_{p, \calP} \leq n/2) = 1 - \P_k(n - a_{p, \calP} \geq n/2).$ Applying Markov's Inequality, we get that $\P_k(n - a_{p, \calP} \geq n/2) \leq \frac{2}{n}\E_k[n - a_{p, \calP}] = \frac{2}{n}(n - \E_k[a_{p, \calP}]) = 2 - \frac{2}{n}\E_k[a_{p, \calP}].$ Overall, this means that $\P_k(a_{p, \calP} > n/2) \geq 1 - (2 - \frac{2}{n}\E_k[a_{p, \calP}]) = -1 + \frac{2}{n}\E_k[a_{p, \calP}] \geq \frac{\Delta2^{1 - t}}{nt}\binom{t - 1}{\floor*{t / 2}}$. We now use the following tight estimation of the central binomial coefficient:
\begin{equation*}
\frac{4^m}{\sqrt{\pi(m + \frac{1}{2})}} \leq \binom{2m}{m} \leq \frac{4^m}{\sqrt{\pi m}}\quad 
\text{(for all $m \geq 1$)}
\end{equation*}
to get that
\begin{gather*}
\P_k(a_{p, \calP} > n/2) \geq  \frac{\Delta2^{1 - t}}{nt}\frac{2^{t - 1}}{\sqrt{\pi(\floor*{t / 2} + \frac{1}{2})}} = \frac{\Delta}{nt\sqrt{\pi(\floor*{t / 2} + \frac{1}{2})}} \\
= \sqrt{\frac{2}{\pi}}\frac{\Delta}{nt^{3/2}} \qedhere
\end{gather*}
\end{proof}

Armed as such, we now give our randomized algorithm in the following theorem.

\begin{theorem} For odd $t$, a winning proposal $p \in B_m$ can be found in expected time $O(n^2t^{5/2}/\Delta)$. If only $\Delta > 0$ is guaranteed, this is $O(n^2t^{7/2})$. However, if each column has more ones than zeros, then $\Delta \geq t$, so the algorithm runs in expected time $O(n^2t^{5/2})$.
\end{theorem}
\begin{proof} Introduce the notation $K = \{\floor*{t / 2} + 1, \ldots, t\}$. We proceed in rounds. In each round we sample proposals $(p_k)_{k \in K}$ such that $|p_k| = k$ uniformly at random. If any of the sampled proposals is winning, then we stop and return that proposal. Otherwise, we proceed to the next round. Each round takes time $O(nt^2)$ to execute, so we are left with bounding the expected number of rounds. From Lemmas \ref{lemma_exists_k} and \ref{lemma_exists_k_implies_a} there is $k^* \in K$ such that $\P_k(b_{p, \calP} > 0) \geq \sqrt{\frac{2}{\pi}}\frac{\Delta}{nt^{3/2}}$, meaning that in each round the sampled proposal $p_{k^*}$ will be winning with at least this probability. As a result, by the expectation of the geometric distribution, the expected number of rounds until $p_{k^*}$ is winning is $O(nt^{3/2}/\Delta),$ which is also an upper-bound on the expected number of rounds. Therefore, our algorithm runs in time $O(nt^2nt^{3/2}/\Delta) = O(n^2t^{7/2}/\Delta),$ as required.
\end{proof}

\paragraph{\textbf{Deterministic Algorithm}}
Recall that a proposal is a mapping $p : [t] \to \B$. A \emph{partial proposal} is a mapping $p : [t] \to \{0, 1, \Qm\}.$ Seeing $p$ as a partial function, the \emph{domain} of $p$ is $\calD_p = \{i \in [t] : p_i \neq \Qm\}.$ Partial proposal $p$ can also be seen as the set $\{i \mapsto p_i : i \in \calD_p\}$. We extend our previous notation $|p|$ consistently to mean $|\{i : p_i = 1\}|.$ The \emph{union} (in the sense of sets) $p_1 \cup p_2$ of two partial proposals $p_1$ and $p_2$ is defined whenever $\calD_{p_1} \cap \calD_{p_2} = \varnothing.$ We say that a partial proposal $p'$ \emph{refines} a partial proposal $p$ if for all $i \in [t]$ we have that either $p_i = p'_i$ or $p_i = \Qm$. Given a partial proposal $p$, one way to refine it is to pick $i \notin \calD_p$ and $b \in \B$ and assign $p_i \gets b$; under the set notation, the refined proposal is written $p \cup \{i \mapsto b\}.$

To aid presentation, for a partial proposal $p^*$ introduce the notation $\P_{p^*}[\cdot] = \P[\cdot \text{ given }p \text{ refines }p^*]$ and $\P_{k, p^*}[\cdot] = \P[\cdot \text{ given }|p| = k \text{ and } p \text{ refines }p^*].$ We define $\E_{p^*}$ and $\E_{k, p^*}$ similarly. The proof of the following lemma is mostly a matter of syntactic manipulation.

\begin{lemma} \label{lemma_compute_expected_partial} For any partial proposal $p^* : [t] \to \{0, 1, \Qm\}$ and any number $|p^*| \leq k \leq t - |\calD_{p^*}|$ the expectation $\E_{k, p^*}[b_{p, \calP}]$ can be computed in polynomial time.
\end{lemma}
\begin{proof}
As usual, by linearity of expectation, $\E_{k, p^*}[b_{p, \calP}]$ can be written as the sum $\sum_{i = 1}^{n}\E_{k, p^*}[Z_i]$, where $Z_i$ is $1$ if $i$ approves $p$, $-1$ if they disapprove it, and $0$ otherwise. Assume that agent $i$'s vote is $v_i$ with $|v_i| = \ell$. Let $B_{+}$ be the set of proposals $p$ refining $p^*$ with $|p| = k$ that $i$ approves of. Likewise, let $B_{-}$ be the set of proposals $p$ refining $p^*$ with $|p| = k$ that $i$ disapproves of. Now, we can write:
\begin{gather*}
\E_{k, p^*}[Z_i] = \binom{t - |\calD_{p^*}|}{k - |p^*|}^{-1}(|B_{+}| - |B_{-}|)
\end{gather*}
Binomial coefficients can be computed in polynomial time, so we are now left with showing how to compute $|B_{+}|$ and $|B_{-}|$ in polynomial time. We begin with $B_{+}.$ First, let us determine the number of matches between $v$ and $p^*$ on issues in $\calD_{p^*},$ which is $m = |\{j \in \calD_{p^*} : v_j = p^*_j\}|.$ Then, let $\alpha = |v \setminus \calD_{p^*}|$ be the number of issues not in $\calD_{p^*}$ which $i$ approves of, and $\beta = t - |\calD_{p^*}| - \alpha$ be the number of issues not in $\calD_{p^*}$ which $i$ disapproves of. For an arbitrary proposal $p$, define $x(p) = |p \setminus \calD_{p^*}|$ and $y(p) = t - |\calD_{p*}| - |p \setminus \calD_{p^*}|$, representing the number of issues not in $\calD_{p^*}$ with a one/zero at the corresponding position in $p$. Note that the total number of proposals $p$ refining $p^{*}$ with $x(p) = x$ and $y(p) = y$ for some numbers $x, y$ is $\binom{\alpha}{x}\binom{\beta}{y},$ assuming that the binomial coefficients are $0$ whenever undefined. These being said, $p \in B_{+}$ holds if and only if $|p^*| + x(p) + y(p) = k$ and $m + x(p) + \alpha + \beta - y(p) > t / 2.$ Therefore, to compute $|B_{+}|$ we can iterate over all pairs $(x, y)$ with $0 \leq x, y \leq \alpha + \beta$ that satisfy $|p^*| + x + y = k$ and $m + x + \alpha + \beta - y > t / 2$ (of which there are polynomially many), adding up the value $\binom{\alpha}{x}\binom{\beta}{y}$ for each. For $B_{-},$ the only difference is that the second condition becomes $m + x + \alpha + \beta - y < t / 2.$
\end{proof}

With this in mind, we can now state and prove our main result.

\begin{theorem} There is a polynomial-time deterministic algorithm that computes a non-losing policy $p \in B_m.$ If $\Delta > 0,$ then the computed policy is winning.
\end{theorem}
\begin{proof}
We only show the ``non-losing'' part of the assertion. For the ``winning'' part. just replace non-negative with positive in the following. 
First, for all values $k$ such that $t / 2 < k \leq t$ we compute the values $\E_k[b_{p, \calP}]$. This can be done in polynomial time by invoking Lemma \ref{lemma_compute_expected_partial} with $p^* = \varnothing$. By Lemma \ref{lemma_exists_k}, for at least one such $k$ the computed expectation will be non-negative, so take $k^*$ to be one such $k$.
Afterwards, the algorithm will build the output policy iteratively, initially starting with $p^* = \varnothing$. At step $1 \leq i \leq t$ of the algorithm, the current proposal $p^*$ will have $\calD_{p*} = \{1, \ldots, i - 1\}$, with the running invariant that $\E_{k^*, p^*}[b_{p, \calP}] \geq 0$. In step $i$, we consider two refinements of $p^*$, namely $p_0 = p^* \cup \{i \mapsto 0\}$ and $p_1 = p^* \cup \{i \mapsto 1\}$. With them, we can write:
\begin{gather*}
    \E_{k^*, p^*}[b_{p, \calP}] = \sum_{j \in \{0, 1\}}\E_{k^*, p_j}[b_{p, \calP}]\P_{k^*, p^*}(p_i = j)
\end{gather*}
From the invariant we know that $\E_{k^*, p^*}[b_{p, \calP}] \geq 0$, and, since $\sum_{j \in \{0, 1\}}\P_{k^*, p^*}(p_i = j) = 1$, we get that for some $j \in \{0, 1\}$ it holds that $\E_{k^*, p_j}[b_{p, \calP}] \geq 0$. Since the two values $\E_{k^*, p_j}[b_{p, \calP}]$ can be computed in polynomial time using Lemma \ref{lemma_compute_expected_partial}, we can thus take $j^*$ to be such that $\E_{k^*, p_{j^*}}[b_{p, \calP}] \geq 0$.
We can now set $p^* \gets p_{j^*}$ and continue with the algorithm. One technical caveat is the situation where $|p_1| > k^*$, in which one expectation is not defined. In this case we just take $j^* = 0$ without computing any expectation. At the end, proposal $p^*$ will be complete and $\E_{k^*, p^*}[b_{p, \calP}] \geq 0$ will hold, meaning that $b_{p^*, \calP} \geq 0,$ so we can output $p^*.$
\end{proof}

One downside of the deterministic algorithm is that it requires high-precision arithmetic to execute. Indeed, the expectations are rational numbers with denominator/numerators on the order $\approx 2^t.$ This makes an efficient implementation tricky, but achievable in polynomial time if we work over the integers with the expectations multiplied by the common denominator. We omit these details. For the randomized algorithm, one might rightfully ask whether it can be optimized by computing $k^*$ in advance and only sampling for it. The answer is that this requires essentially the same machinery as the deterministic algorithm, making the time complexity less attractive.

Finally, note that our polynomial algorithms output non-losing proposals with at least $\floor*{t / 2} + 1$ agreements to \IWM, and deciding whether at least $\floor*{t / 2} + 2$ is possible is NP-hard, by Theorem \ref{th_main_np_hardness_result}. Perhaps counterintuitively, this does not mean that our algorithms output proposals with exactly $\floor*{t / 2} + 1$ agreements, as such proposals may actually not always exist (see appendix \ref{sect_appendix} for an example).

\section{Conclusions and Future Work}

We studied the problem of determining a policy minimizing the distance to issue-wise majority, while at the same time surviving the final vote of the assembly to approve it as the outcome. In essence, our results establish a tight dichotomy:~distance at most $\floor*{(t - 1) / 2}$ can always be achieved in polynomial time, while deciding whether a better distance is possible is NP-hard. It would be natural to reexamine our results through the lens of general judgement aggregation (i.e., dependent issues) to identify which extensions are possible. Moreover, we assumed that voters weigh all of the issues equally, but it would be of increased significance to study a setup where voters give issue-importance scores together with their ballots. Additionally, some voters might be pickier than others, requiring significantly more than $50\%$ agreement with their preferences in order to support a proposal. It would be interesting to also incorporate such behaviour into our model. The approval ``supports/opposes'' paradigm can also be replaced by other voting mechanisms, perhaps also defined in terms of the Hamming distance, but without a fixed approval ``threshold''. Finally, the setup of non-binary issues would also be interesting to investigate.

The Ostrogorski paradox generalizes Anscombe's. In particular, in Anscombe's paradox, issue-wise majority loses against the opposite (issue-wise minority) proposal, while in Ostrogorski's issue-wise majority loses against an arbitrary proposal. In fact, Ostrogorski's paradox can be seen to be equivalent to the Condorcet paradox restricted to our setup (see, e.g., \cite{paradoxes_explained}). An algorithmic version of Ostrogorski's paradox asks the following:~given the voter preferences, is there a policy defeating all other policies when faced in 1-vs-1 match-ups? It is known that this policy, if it exists, has to be issue-wise majority \cite{paradoxes_explained}, but determining whether this is the case seems interesting and computationally non-trivial.


\begin{acks}
We thank Robin Fritsch and Edith Elkind for the many fruitful discussions regarding this work. We additionally thank Robin Fritsch for pointing out a case where a non-losing proposal with $\floor*{t / 2} + 1$ agreements with issue-wise majority might not exist. We thank the anonymous reviewers for their constructive feedback and useful suggestions contributing to improving this paper.
\end{acks}



\bibliographystyle{ACM-Reference-Format} 
\bibliography{majority}


\appendix
\section{The Space of Solutions is Disconnected} \label{sect_appendix}

Our two polynomial algorithms are not completely fulfilling, since their correctness ultimately relies on probabilistic techniques. We leave it open whether a fully-constructive algorithm exists. One possible approach one might be tempted to consider is to look at the subgraph induced by non-losing proposals in the hypercube graph $H_t$ with vertex set $\B^t$ and edges between proposals differing in exactly one issue. If the graph is sufficiently connected, then performing a walk directed by $|p|$ might be able to find a proposal with $|p| \geq \floor*{t / 2} + 1.$ Unfortunately, as we show next, the graph might actually have no edges, so a more refined approach would be needed. In particular, we exhibit an instance where a proposal is winning if and only if it agrees with issue-wise majority in an odd number of issues, a combinatorial curiosity which might be of independent interest. For convenience, the result is stated with voters/issues numbered starting with $0.$

\begin{figure}[t]
\centering
\begin{tabular}{c|ccccc}
      & 0 & 1 & 2 & 3 & 4 \\
      \hline
    0 & \tikzmarknode{ex_1_1}{1} & \tikzmarknode{ex_1_2}{1} & \tikzmarknode{ex_1_3}{1} & \tikzmarknode{ex_1_4}{0} & \tikzmarknode{ex_1_5}{0} \\
    1 & \tikzmarknode{ex_2_1}{0} & \tikzmarknode{ex_2_2}{1} & \tikzmarknode{ex_2_3}{1} & \tikzmarknode{ex_2_4}{1} & \tikzmarknode{ex_2_5}{0} \\
    2 & \tikzmarknode{ex_3_1}{0} & \tikzmarknode{ex_3_2}{0} & \tikzmarknode{ex_3_3}{1} & \tikzmarknode{ex_3_4}{1} & \tikzmarknode{ex_3_5}{1} \\
    3 & \tikzmarknode{ex_4_1}{1} & \tikzmarknode{ex_4_2}{0} & \tikzmarknode{ex_4_3}{0} & \tikzmarknode{ex_4_4}{1} & \tikzmarknode{ex_4_5}{1} \\
    4 & \tikzmarknode{ex_5_1}{1} & \tikzmarknode{ex_5_2}{1} & \tikzmarknode{ex_5_3}{0} & \tikzmarknode{ex_5_4}{0} & \tikzmarknode{ex_5_5}{1}
\end{tabular}
\caption{Construction from Theorem \ref{theorem_counterexample} for $t = 5.$ Here all policies $p$ with $|p| \in \{1, 3, 5\}$ are winning, while all with $|p| \in \{0, 2, 4\}$ are losing.}
\label{fig:counterexample}
\Description{Construction used to get an empty graph.}
\begin{tikzpicture}[overlay,remember picture, shorten >=-3pt, shorten <= -3pt]
\drawline{ex_1_1}{ex_1_3}{\cone}
\drawline{ex_2_2}{ex_2_4}{\cone}
\drawline{ex_3_3}{ex_3_5}{\cone}
\drawline{ex_4_1}{ex_4_1}{\cone}
\drawline{ex_4_4}{ex_4_5}{\cone}
\drawline{ex_5_1}{ex_5_2}{\cone}
\drawline{ex_5_5}{ex_5_5}{\cone}
\drawline{ex_1_4}{ex_1_5}{\czero}
\drawline{ex_2_1}{ex_2_1}{\czero}
\drawline{ex_2_5}{ex_2_5}{\czero}
\drawline{ex_3_1}{ex_3_2}{\czero}
\drawline{ex_4_2}{ex_4_3}{\czero}
\drawline{ex_5_3}{ex_5_4}{\czero}
\end{tikzpicture}
\end{figure}

\begin{theorem} \label{theorem_counterexample}Consider a judgements matrix $\calP$ with odd $t = n,$ where voter $i$ approves of issues $\{i, i + 1, \ldots, i + \floor*{t / 2}\}$, addition of indices being performed modulo $t$. Then, a proposal $p \in \B^t$ wins if and only if $|p|$ is odd. Figure \ref{fig:counterexample} illustrates this construction for $t = 5$.
\end{theorem}

\begin{proof} For convenience, we will look at proposals as vectors in $\{\pm 1\}^t.$ Let $p \in \{\pm 1\}^t$ be an arbitrary proposal. For brevity, we write $b_i = b_{v_i, p}.$ Note that $b_i = \sum_{j \in v_i}p_j - \sum_{j \notin v_i}p_j.$ Observe that $b_i + b_{i + \floor*{t / 2} + 1}$ equals 
\begin{gather*}
     \sum_{j = i}^{i + \floor*{t / 2}}p_j - \sum_{j = \floor*{t / 2} + 1}^{i - 1} p_j + \sum_{j = i + \floor*{t / 2} + 1}^{i}p_j - \sum_{j = i + 1}^{i + \floor*{t / 2}} p_j = 2p_i
\end{gather*}
For ease of writing, write $g = \floor*{t / 2} + 1.$
Since $t$ is odd, $g$ is a generator in $\Z_t$; i.e.~$\{ig\}_{0 \leq i < t} = \Z_t.$ From the above we have that $b_{ig} + b_{(i + 1)g} = 2p_{ig}.$ Since $p_{ig} \in \{\pm 1\}$, it follows that $2p_{ig} \equiv 2 \Mod{4}.$ Since $b_{i}$ is odd for all $0 \leq i < t,$ the previous means that $b_{ig} \equiv b_{(i + 1)g} \Mod{4}.$ Since $g$ generates $\Z_t,$ this means that there is a number $c \in \{\pm 1\}$ such that $b_{i} \equiv c \Mod{4}.$ Observe that
\begin{gather*}
\sum_{i = 0}^{t - 1} b_i = \frac{1}{2}\left(\sum_{i = 0}^{t - 1} b_i + \sum_{i = 0}^{t - 1} b_{i + \floor*{t / 2} + 1}\right) = \frac{1}{2}\sum_{i = 0}^{t - 1}(b_i + b_{i + \floor*{t / 2} + 1}) = \sum_{i = 0}^{t - 1}p_i
\end{gather*}
This implies that $tc \equiv 2|p| - t \Mod{4},$ from which $t(c + 1) \equiv 2|p| \Mod{4},$ meaning that if $|p|$ is odd then $c = 1$ and if $|p|$ is even then $c = -1.$

Now, since $b_{ig} + b_{(i + 1)g} = 2p_{ig}$ and the two summands are odd, note that $b_{ig}$ and $b_{(i + 1)g}$ have opposite signs unless either they are both $1$ or both $-1$. Let $s_i$ denote the sign of $b_i$ for all $0 \leq i < t$. Observe that $b_{p, \calP} = \sum_{i = 0}^{t - 1}s_i$. Now, write $b_{p, \calP}$ as follows:
\begin{gather*}
b_{p, \calP} = \frac{1}{2}\left(\sum_{i = 0}^{t - 1}s_{ig} + \sum_{i = 0}^{t - 1}s_{(i + 1)g}\right) = \frac{1}{2}\sum_{i = 0}^{t - 1}(s_{ig} + s_{(i + 1)g})
\end{gather*}
The last expression equals $c|S|,$ where $S \subseteq \Z_t$ is the set of indices $i$ such that $s_{ig} = s_{(i + 1)g},$ as in all other cases the signs are opposite. Since $c$ is $\pm 1$ based on whether $|p|$ is even or odd, we get the conclusion.
\end{proof}

Observe that for $t = 7$ our construction exhibits the property that a non-losing proposal with $|p| = \floor*{t / 2} + 1$ does not exist, while non-losing proposals with higher $|p|$ do. This shows that, indeed, our polynomial algorithms are not guaranteed to output a proposal with $|p| = \floor*{t / 2} + 1,$ despite it being NP-hard to decide whether $|p| > \floor*{t / 2} + 1$ is possible, expanding on our earlier remark.

\end{document}